\documentclass{article}
\usepackage[T1]{fontenc}
\usepackage{ae,aecompl}
\usepackage[utf8x]{inputenc}
\usepackage{fullpage}
\usepackage[square,sort]{natbib}
\usepackage{times}
\usepackage{amsfonts}
\usepackage{amsmath}
\usepackage{amssymb}
\usepackage{amsthm}
\usepackage{mathtools}
\usepackage{graphicx}
\usepackage{subfig}
\usepackage{multirow}
\usepackage[ruled,linesnumbered]{algorithm2e}

\newcommand{\Sam}{{\cal S}}
\newcommand{\Db}{{\cal D}}
\newcommand{\Tab}{{\cal T}}

\newtheorem{corol}{Corollary}

\newtheorem{lemma}{Lemma}

\newtheorem{thm}{Theorem}

\theoremstyle{definition}
\newtheorem{defn}{Definition}

\DeclareMathOperator{\op}{op}
\DeclareMathOperator{\eqop}{eqop}
\DeclareMathOperator{\bool}{bool}

\begin{document}
\title{The VC-Dimension of SQL Queries and Selectivity Estimation Through
Sampling\thanks{Work was supported in part by NSF award IIS-0905553.}}

\author{Matteo Riondato\thanks{Contacth autor.} \and Mert Akdere \and U\v gur \c
Cetintemel \and Stanley B. Zdonik \and Eli Upfal \\ 
Department of Computer Science, Brown University, Providence, RI \\
\texttt{\{matteo,makdere,ugur,sbz,eli\}@cs.brown.edu}
}
\maketitle
\begin{abstract}
We develop a novel method, based on the statistical concept of the
\emph{Vapnik-Chervonenkis dimension}, to evaluate the selectivity (output
cardinality) of SQL queries -- a crucial step in optimizing the execution of
large scale database and data-mining operations.  The major theoretical
contribution of this work, which is of independent interest, is an explicit
bound to the VC-dimension of a range space defined by all possible outcomes of a
collection (class) of queries. We prove that the VC-dimension is a function of
the maximum number of Boolean operations in the selection predicate and of the
maximum number of select and join operations in any individual query in the
collection, but it is neither a function of the number of queries in the
collection nor of the size (number of tuples) of the database.  We leverage on
this result and develop a method that, given a class of queries, builds a
concise random sample of a database, such that with high probability the
execution of \emph{any} query in the class on the sample provides an accurate
estimate for the selectivity of the query on the original large database. The
error probability holds \emph{simultaneously} for the selectivity estimates of
\emph{all} queries in the collection, thus the same sample can be used to
evaluate the selectivity of multiple queries, and the sample needs to be
refreshed only following major changes in the database. The sample
representation computed by our method is typically sufficiently small to be
stored in main memory. We present extensive experimental results, validating our
theoretical analysis and demonstrating the advantage of our technique when
compared to complex selectivity estimation techniques used in PostgreSQL and the
Microsoft SQL Server.
\end{abstract}

\section{Introduction}\label{sec:intro}
As advances in technology allow for the collection and storage of vast
databases, there is a growing need for advanced machine learning techniques for
speeding up the execution of queries on such large datasets. In this work we
focus on the fundamental task of estimating the selectivity, or output size, of
a database query, which is a crucial step in a number of query processing tasks
such as execution plan optimization and resource allocation in parallel and
distributed databases. The task of efficiently obtaining such accurate estimates
has been extensively studied in previous work with solutions ranging from
storage of pre-computed statistics on the distribution of values in the tables,
to online sampling of the databases, and to combinations of the two
approaches~\citep{LiptonN95,LiptonNS90,HaasS92,HouOD91,HaasS95,GangulyGMS96,GantiLR00,GibbonsM98,HouOT88,LarsonLZZ07,PoosalaI97}.
Histograms, simple yet powerful statistics of the data in the tables, are the most
commonly used solution in practice, thanks to their computational and space
efficiency. However, there is an inherent limitation to the accuracy of this
approach when estimating the selectivity of queries that involve either multiple
tables/columns or correlated data. Running the query on freshly sampled data
gives more accurate estimates at the cost of delaying the execution of the query
while collecting random samples from a disk or other large storage medium and
then performing the analysis itself. This approach is therefore usually more
expensive than a histogram lookup. Our goal in this work is to exploit both the
computational efficiency of using pre-collected data and the provable accuracy
of estimates obtained by running a query on a properly sized random sample of
the database.

We apply the statistical concept of VC-dimension~\citep{VapnikC71} to develop and
analyze a novel technique to generate accurate estimates of query
selectivity. Roughly speaking, the VC-dimension of a collection of indicator
functions (hypotheses) is a measure of its complexity or expressiveness (see
Sect.~\ref{sec:vcdim} for formal definitions). A major theoretical contribution
of this work, which is of independent interest, is an explicit bound to the
VC-dimension of any class of queries, viewed as indicator functions on the
Cartesian product of the database tables. In particular, we show that the
VC-dimension of a class of queries is a function of the maximum number of
Boolean, select and join operations in any query in the class, but it is not
a function of the number of different queries in the class. By adapting a
fundamental result from the VC-dimension theory to the database setting, we
develop a method that for any class of queries, defined by its VC-dimension,
builds a concise sample of the database, such that with high probability, the
execution of \emph{any} query in the class on the sample provides an accurate
estimate for the selectivity of the query on the original large database. The
error probability holds \emph{simultaneously} for the selectivity estimate of
\emph{all} queries in the collection, thus the same sample can be used to
evaluate the selectivity of multiple queries, and the sample needs to be
refreshed only following major changes in the database. The size of the sample
does not depend on the size (number of tuples) in the database, just on the
complexity of the class of queries we plan to run, measured by its VC-dimension.
Both the analysis and the experimental results show that accurate selectivity
estimates can be obtained using a sample of a surprising small size (see
Table~\ref{tab:samplesize} for concrete values), which can then reside in main
memory, with the net result of a significant speedup in the execution of
queries on the sample. 

A technical difficulty in applying the VC-dimension results to the database
setting is that they assumes the availability of a uniform sample of the
Cartesian product of all the tables, while in practice it is more efficient to
store a sample of each table separately and run the queries on the Cartesian
product of the samples, which has a different distribution than a sample of the
Cartesian product of the tables. We develop an efficient  procedure for
constructing a sample that circumvents this problem (see
Sect.~\ref{sec:applications}).

We present extensive experimental results that validate our theoretical analysis
and demonstrate the advantage of our technique when compared to complex
selectivity estimation techniques used in PostgreSQL and the Microsoft SQL
Server. The main advantage of our method is that it gives provably accurate
predictions for the selectivities of all queries with up to a given complexity
(VC-dimension) specified by the user before creating the sample, while
techniques like multidimensional histograms or join synopses are accurate only
for the queries for which they are built.

Note that we are only concerned with estimating the selectivity of a query, not
with approximating the query answer using a sample of the database
(Das~\citeyearpar{Das09} presents a survey of the possible solutions to this latter
task). %

\paragraph{Outline.} The rest of the paper is organized as follows. We review
the relevant previous work in Sect.~\ref{sec:prevwork}. In
Sect.~\ref{sec:prelim} we formulate the problem and introduce the
Vapnik-Chervonenkis dimension and the related tools we use in developing our
results. Our main analytical contribution, a bound on the VC dimension  of class of queries is presented in Sect.~\ref{sec:vcdimqueries}.  The application of these results for selectivity estimation is
given in Sect.~\ref{sec:applications}. Experiments are presented in
Sect.~\ref{sec:experiments}. 

\section{Related Work}\label{sec:prevwork}
Methods to estimate the selectivity (or cardinality of the output) of queries
have been extensively studied in the database literature primarily due to the
importance of this task  to query plan optimization and resource allocation. A
variety of approaches have been explored, ranging from the use of sampling, both
online and offline, to the pre-computation of different statistics such as
histograms, to the application of methods from machine
learning~\citep{ChenMM90,HarangsriNS97}, data mining~\citep{GryzL04},
optimization~\citep{ChaudhuriDN07,MarklHKMST07}, and probabilistic
modeling~\citep{GetoorTK01,ReS10}.

The use of sampling for selectivity estimation has been studied mainly in the
context of online sampling~\citep{LiptonNS90,LiptonN95}, where a sample is
obtained, one tuple at a time, after the arrival of a query and it used only
to evaluate the selectivity of that query and then discarded. Sampling at random
from a large database residing
on disk is an expensive operation~\citep{Olken93,BrownH06,GemullaLH06}, and in
some cases sampling for an accurate cardinality estimate is not significantly
faster than full execution of the query~\citep{HaasNSS93,HaasNS94}.

A variety of sampling and statistical analysis techniques has been tested to 
improve the efficiency of the sampling procedures and in particular
to identify early stopping conditions. These include sequential sampling
analysis~\citep{HouOD91,HaasS92}, keeping additional statistics to improve the
estimation~\citep{HaasS95}, labelling the tuples and using label-dependent
estimation procedures~\citep{GangulyGMS96}, or applying the cumulative distribution
function inversion procedure~\citep{WuAE01}. Some work also looked at nonuniform
sampling~\citep{BabcockCD03,EstanN06} and stratified
sampling~\citep{ChaudhuriDN07,JoshiJ08}. Despite all these relevant
contributions, online sampling is still considered too expensive for most
applications. An offline sampling approach was explored by
Ngu~et~al.~\citeyearpar{NguHS04}, who used systematic sampling (requiring the tuples in
a
table to be sorted according to one of the attributes) with a
sample size dependent on the number of tuples in the table. The paper does
not give any explicit guarantee on the accuracy of their predictions.
Chaudhuri~et~al.~\citeyearpar{ChaudhuriDN07} present an approach which uses
optimization techniques to identify suitable strata before sampling. The
obtained sample is such that the mean square error in estimating the selectivity
of queries belonging to a given workload is minimized, but there is no quality
guarantee on the maximum error. Haas~\citeyearpar{Haas96} developed Hoeffding
inequalities to bound the probability that the selectivity of a query estimated
from a sample deviates more than a given amount from its expectation.
However, to estimate the selectivity for multiple queries and obtain a given
level accuracy for all of them, simultaneous statistical inference techniques
like the union bound should be used, which are known to be overly conservative
when the number of queries is large~\citep{Miller81}. On the contrary, our result
will hold simultaneously for \emph{all} queries within a given complexity (VC
dimension). 

A technical problem arises when combining join operations and sampling. As
pointed out by Chaudhuri~et~al.~\citeyearpar{ChaudhuriMN99}, the Cartesian
product of uniform samples of a number of tables is different from a uniform sample of the
Cartesian product of those tables. Furthermore, given a size $s$, it is
impossible to a priori determine two sample sizes $s_1$ and $s_2$ such that
uniform samples of these sizes from the two tables will give, when joined
together along a common column, a sample of the join table of size $s$. In
Sect.~\ref{sec:applications} we explain why only the first issue is of concern
for us and how we circumvent it.

In practice most database systems use pre-computed statistics to predict query
selectivity~\citep{HouOT88,GibbonsM98,GantiLR00,JinGJA06,LarsonLZZ07}, with
histograms being the most commonly used representation. The construction,
maintenance, and use of histograms were thoroughly examined in the
literature~\citep{JagadishKMPSS98,IoannidisP95,MatiasVW98,PoosalaHIS96}, with
both theoretical and experimental results. In particular
Chaudhuri~et~al.~\citeyearpar{ChaudhuriMN98} rigorously evaluated the size of the
sample needed for
building a histogram providing good estimates for the selectivities of a large
group of (select only, in their case) queries.
Kaushik~et~al.~\citeyearpar{KaushikNRC05} extensively compared histograms and sampling
from a
space complexity point of view, although their sample-based estimator did not
offer a uniform probabilistic guarantee over a set of queries and they only
consider the case of foreign-key equijoins. We address both these points in our
work. Although very efficient in terms of storage needs and query time, the
quality of estimates through histograms is inherently limited for complex
queries because of two major drawbacks in the use of histograms: intra-bucket
uniformity assumption (i.e., assuming a uniform distribution for the frequencies
of values in the same bucket) and inter-column independence assumption (i.e.,
assuming no correlation between the values in different columns of the same or
of different tables).  Different authors suggested solutions to improve the
estimation of selectivity without making the above
assumptions~\citep{BrunoC04,Dobra05,PoosalaI97,WangVI97,WangS03}. Among these
solutions, the use of multidimensional
histograms~\citep{BrunoCG01,PoosalaI97,SrivastavaHMKT,WangS03} seems the most
practical. Nevertheless, these techniques are not widespread due to the extra
memory and computational costs in their implementation.

Efficient and practical techniques for drawing random samples from a database
and for updating  the sample when the underlying tables evolve have been extensively
analyzed in the database
literature~\citep{BrownH06,GemullaLH06,GemullaLH07,HaasK04,JermainePA04}.

The {\em Vapnik-Chervonenkis dimension} was first introduced in a seminal
article~\citep{VapnikC71} on the convergence of probability distributions, but it
was only with the work of Haussler and Welzl~\citeyearpar{HausslerW86} and 
Blumer~et~al.~\citeyearpar{BlumerEHW89} that it was applied to computational sampling
and learning theory. Since then, VC-dimension has encountered
enormous success and application in the fields of computational
geometry~\citep{Chazelle00,Matousek02} and machine learning~\citep{AnthonyB99} but
its use in system-related fields is not as widespread. In the database
literature, it was used in the context of constraint databases to compute good
approximations of aggregate operators~\citep{BenediktL02}. VC-dimension-related
results were also recently applied in the field of database privacy by Blum,
Ligett, and Roth~\citeyearpar{BlumLR08} to show a bound on the number of queries needed
for an attacker to learn a private concept in a database.
Gross-Amblard~\citeyearpar{Gross11} showed that content with unbounded VC-dimension can
not be watermarked for privacy purposes. 

To the best of our knowledge, our work is the first to provide explicit bounds
on the VC-dimension of queries and to apply the results to query selectivity
estimation.

\section{Preliminaries}\label{sec:prelim}
We consider a database $\Db$ of $k$ tables $\Tab_1,\cdots,\Tab_k$. We denote a
column $C$ of a table $\Tab$ as $\Tab.C$ and, for a tuple $t\in\Tab$, the value
of $t$ in the column $C$ as $t.C$. We denote the domain of the values that can
appear in a column $\Tab.C$ as $D(\Tab.C)$. Our focus is on queries that combine
select and join operations, defined as follows. We do not take projection
operations into consideration because their selectivities have no impact on
query optimization.

\begin{defn}\label{def:selectquery}
  Given a table $\Tab$ with columns $\Tab.C_1,\cdots,\Tab.C_\ell$, a
  \emph{selection query} $q$ on $\Tab$ is an operation which returns a subset
  $S$ of the tuples of $\Tab$ such that a tuple $t$ of $\Tab$ belongs to $S$ if
  and only if the values in $t$ satisfy a condition $\mathcal{C}$
  (the \emph{selection predicate}) expressed by $q$. In full
  generality, $\mathcal{C}$ is the Boolean combination of clauses of the form
  $\Tab.C_i \op a_i$, where $\Tab.C_i$ is a column of $\Tab$, ``$\op$'' is one
  of $\{<,>,\ge,\le,=,\neq\}$ and $a_i$ is an element of the domain of
  $\Tab.C_i$.
\end{defn}

We assume that all $D(\Tab.C_i)$ are such that it is possible to build total
order relations on them. This assumptions does not exclude categorical domains
from our discussion, because the only meaningful values for ``$\op$'' for such
domains are ``$=$'' and ``$\neq$'', so we can just assume an arbitrary but fixed
order for the categories in the domain.

\begin{defn}\label{def:joinquery}
  Given two tables $\Tab_1$ and $\Tab_2$, a $\emph{join query}$ $q$ on a common
  column $C$ (i.e. a column present both in $\Tab_1$ and $\Tab_2$) is an
  operation which returns a subset of the Cartesian product of the tuples in
  $\Tab_1$ and $\Tab_2$. The returned subset is defined as the set
  \[
  \{(t_1,t_2) ~:~ t_1\in\Tab_1, t_2\in\Tab_2, \mbox{ s.t. } t_1.C \op t_2.C \}\]
  where ``$\op$'' is one of $\{<,>,\ge,\le,=,\neq\}$.
\end{defn}

Our definition of a join query is basically equivalent to that of a
\emph{theta-join}~\citep[Sect.5.2.7]{GarciaMolinaUW02}, with the limitation that
the join condition $\mathcal{C}$ can only contain a single clause, i.e. a single
condition on the relationship of the values in the shared column $C$ and only
involve the operators $\{<,>,\ge,\le,=,\neq\}$ (with their meaning on $D(C)$).
The pairs of tuples composing the output of the join in our definition have a
one-to-one correspondence with the tuples in the output of the corresponding
theta-join.

\begin{defn}\label{def:general query}
  Given a set of $\ell$ tables $\Tab_1,\cdots,\Tab_\ell$, a
  \emph{combination of select and join operations} is a query returning a subset
  of the Cartesian product of the tuples in the sets $S_1,\cdots,S_\ell$, where
  $S_i$ is the output of a selection query on $\Tab_i$. The returned set is
  defined by the selection queries and by a set of join queries on
  $S_1,\dots,S_\ell$.
\end{defn}

\begin{defn}\label{def:exectree}
  Given a query $q$, a \emph{query plan} for $q$ is a directed binary tree
  $T_q$ whose nodes are the elementary (i.e. select or join) operations into
  which $q$ can be decomposed. There is an edge from a node $a$ to a node $b$ if
  the output of $a$ is used as an input to $b$. The operations on the leaves of
  the tree use one or two tables of the database as input. The output of the
  operation in the root node of the tree is the output of the query.
\end{defn}

It follows from the definition of a combination of select and join operations
that a query may conform with multiple query plans. Nevertheless, for all the
queries we defined there is (at least) one query plan such that all select
operations are in the leaves and internal nodes are join
nodes~\citep{GarciaMolinaUW02}. To derive our results, we use these specific
query plans.

Two crucial definitions that we use throughout the work are the 
\emph{cardinality} of the output of a query and the equivalent concept of
\emph{selectivity} of a query.

\begin{defn}\label{def:cardsel}
  Given a query $q$ and a database $\Db$, the \emph{cardinality} of its output
  is the number of elements (tuples if $q$ is a selection queries, pairs of
  tuples if $q$ is a join query, and $\ell$-uples of tuples for combinations of
  join and select involving $\ell$ tables) in its output, when run on $\Db$. The
  \emph{selectivity} $\sigma(q)$ of $q$ is the ratio between its cardinality and
  the product of the sizes of its input tables.
 \end{defn}

Our goal is to store a succint representation (sample) $\Sam$ of
the database $\Db$ such that an execution of a query on the sample $\Sam$ will
provide an accurate estimate for the selectivity of
each operation in the query plan when executed on the database $\Db$. 

\subsection{VC-Dimension}\label{sec:vcdim}
The Vapnik-Chernovenkis (VC) Dimension of a family of indicator functions (or
equivalently a family of subsets) on a space of points is a measure of
the complexity or expressiveness of  set of functions in that structure~\citep{VapnikC71}. A finite bound
on the VC-dimension of a structure implies a bound on the number of random
samples required for approximately learning that structure. We outline here some
basic definitions and results and their adaptation to the specific setting of
queries. We refer the reader to the works of Alon and
Spencer~\citeyearpar[Sect.~14.4]{AlonS08}, Chazelle~\citeyearpar[Chap.~4]{Chazelle00}, and
Vapnik~\citeyearpar{Vapnik99} for an in-depth discussion of the VC-dimension theory.

VC-dimension is defined on {\em range spaces}:

\begin{defn}\label{defn:rangespace}
  A {\em range space} is a pair $(X,R)$ where $X$ is a (finite or infinite) set
  and $R$ is a (finite or infinite) family of subsets of $X$. The members of $X$
  are called {\em points} and those of $R$ are called {\em ranges}.
\end{defn}

In our setting, for a class of select queries $Q$ on a table $\Tab$, $X$ is the
set of all tuples in the input table, and $R$ the family of the outputs (as sets
of tuples) of the queries in $Q$ when run on 
$\Tab$. For a class
$Q$ of queries combining select and join operations, $X$ is the Cartesian
product of the associated tables and $R$ is the family of outcomes of queries in
$Q$, seen as $\ell$-uples of tuples, if $\ell$ tables are involved in the
queries of $Q$.  When the context is clear we identify the family $R$ with a
class of queries.

To define the VC-dimension of a range space we consider the projection of the
ranges into a set of points:

\begin{defn}\label{defn:proj}
  Let $(X,R)$ be a range space and $A\subset X$. The {\em projection} of $R$ on
  $A$ is defined as $P_R(A)=\{r\cap A ~:~ r\in R\}$.
\end{defn}

A set is said to be shattered if all its subsets are defined by the range space:

\begin{defn}\label{defn:shatter}
  Let $(X,R)$ be a range space and $A\subset X$. If $|P_R(A)|=2^A$, then $A$ is
  said to be {\em shattered by $R$}.
\end{defn}

The VC-dimension of a range space is the cardinality of the largest set
shattered by the space:

\begin{defn}\label{defn:VCdim}
  Let $S=(X,R)$ be a range space. The {\em Vapnik-Chervonenkis} dimension (or
  {\em VC-dimension}) of $S$, denoted as $VC(S)$ is the maximum cardinality of a
  shattered subset of $X$. If there are arbitrary large shattered subsets, then
  $VC(S)=\infty$.
\end{defn}

When the ranges represent all the possible outputs of queries in a class $Q$
applied to database tables $\Db $, the VC-dimension of the range space is the maximum
number of tuples such that any subset of them is the output of a query in $Q$.

Note that a range space $(X,R)$ with an arbitrary large set of points $X$ and
an arbitrary large family of ranges $R$ can have a bounded VC-dimension. A simple
example is the family of intervals in $[0,1]$ (i.e. $X$ is all the points in
$[0,1]$ and $R$ all the intervals $[a,b]$, such that $0\leq a\leq b\leq 1$). Let
$A=\{x,y,z\}$ be the set of three points $0<x<y<z<1$. No interval in $R$ can
define the subset $\{x,z\}$ so the VC-dimension of this range space is $< 3$.
This observation is generalized in the following result~\citep[Lemma
10.3.1]{Matousek02}:

\begin{lemma}\label{lem:matousek}
  The VC-Dimension of the range space $(\mathbb{R}^d, X)$, where $X$ is the set
  of all half-spaces in $\mathbb{R}^d$ equals $d+1$.
\end{lemma}

The main application of VC-dimension in statistics and learning theory is its
relation to the minimum size sample needed for approximate learning of a
set of indicator functions or hypothesis. 

\begin{defn}\label{defn:eapprox}
  Let $(X,R)$ be a range space and let $A$ be a finite subset of $X$. 
  \begin{enumerate}
  \item 
  For
  $0<\varepsilon<1$, a subset $B\subset A$ is an
  $\varepsilon${\em-approximation}
  for $A$ in $(X,R)$  if $\forall r\in R$, we have  $\left|\frac{|A\cap
  r|}{|A|}-\frac{|B\cap r|}{|B|}\right| \leq \varepsilon$.
   \item For
  $0<p,\varepsilon<1$, a subset $B\subset A$ is a \emph{relative}
  $(p,\varepsilon)$\emph{-approximation}
  for $A$ in $(X,R)$ if for any range $r\in R$ such that 
  $\frac{|A\cap r|}{|A|}\geq p$ we have $\left|\frac{|A\cap r|}{|A|}-\frac{|B\cap
  r|}{|B|}\right| \leq \varepsilon\frac{|A\cap r|}{|A|}$ and for any range $r\in R$ such that 
  $\frac{|A\cap r|}{|A|}< p$ we have $\frac{|B\cap r|}{|B|} \leq (1+\varepsilon)p$.
\end{enumerate}
\end{defn}

An $\varepsilon$-approximation (resp.
a relative $(p,\varepsilon)$-approximation) can be probabilistically constructed by sampling the point
space~\citep{VapnikC71,LiLS01,HarPS11}.

\begin{thm}\label{thm:eapprox}
Let $(X,R)$ be a range space of VC-dimension at most $d$,  $A$ a finite subset of $X$, and $B\subset
  A$ a random sample of $A$ of cardinality $s$. 
  \begin{enumerate}
\item
  There is a positive constant $c$ such that for any 
  $0<\varepsilon,\delta<1$  and  
  \begin{equation}\label{eq:eapprox}
  s\ge\min\left\{|A|,\frac{c}{\varepsilon^2}\left(d+\log\frac{1}{\delta}\right)\right\},
  \end{equation}  $B$ is an $\varepsilon$-approximation for $A$ with probability at least $1-\delta$.
\item
There is a positive constant $c'$ such that for any $0<p<1$ and
  $$s\ge\min\left\{|A|,\frac{c'}{\varepsilon^2p}\left(d\log{\frac{1}{p}}+\log\frac{1}{\delta}\right)\right\}$$
  $B$ is a relative
  $(p,\varepsilon)$-approximation for $A$ with probability at least $1-\delta$.
  \end{enumerate}
\end{thm}

L\"offler and Phillips~\citeyearpar{LofflerP09} showed experimentally that the constant
$c$ is approximately $0.5$. It is also interesting to note that an
$\varepsilon$-approximation of size
$O(\frac{d}{\varepsilon^2}\log{\frac{d}{\varepsilon}})$ can be built
\emph{deterministically} in time
$O(d^{3d}(\frac{1}{\varepsilon^2}\log{\frac{d}{\varepsilon}})^d|X|)$~\citep{Chazelle00}.

In Sect.~\ref{sec:applications} we use an
$\varepsilon$-approximation (or a relative $(p,\varepsilon)$-approximation)  to
compute good estimates of the selectivities of all queries in $Q$. We obtain a
small approximation set through probabilistic construction.
The challenge in applying Thm.~\ref{thm:eapprox} to our setting is computing the
VC-dimension of a range space defined by a class of queries. 
We state here a few fundamental results that will be used in our analysis. 

First, it is
clear that if the VC-dimension of a range space $(X,R)$ is $d$ then $2^d\leq
|R|$ since all subsets of some set of size $d$ are defined by members of $R$.
Next we define for integers $n>0$ and $d>0$ the {\em growth function}
$g(d,n)$ as
  \[
  g(d,n)=\sum_{i=0}^d\binom{n}{i} < n^d .\]

The growth function is used in the following results~\citep[Sect.~14.4]{AlonS08}.
\begin{lemma}[Sauer's Lemma]\label{lem:sauer}
  If $(X,R)$ is a range space of VC-dimension $d$ with $|X|=n$ points, then
  $|R|\le g(d,n) $.
\end{lemma}

\begin{corol}\label{corol:SauerProj}
  If $(X,R)$ is a range space of VC-dimension $d$, then for every finite
  $A\subset X$, $|P_R(A)|\le g(d,|A|)$.
\end{corol}

Our analysis in Sect.~\ref{sec:vcdimqueries} uses the following bound which is an extension
of~\citep[Corol.~14.4.3]{AlonS08} to arbitrary combinations of set operations.

\begin{lemma}\label{lem:genboolcomp}
   Let $(X,R)$ be a range space of VC-dimension $d\ge 2$ and let $(X,R_h)$ be the
  range space on $X$ in which $R_h$ include all possible combinations of 
    union and intersections of $h$ members of $R$. Then $VC(X,R_h)\le
  3dh\log(dh)$.
\end{lemma}

\begin{proof}
  Let $A$ be an arbitrary subset of cardinality $n$ of $X$. We have $|P_R(A)|\le
  g(d,n) \le n^d$. There are 
  \[
  \binom{|P_R(A)|}{h}\le \binom{g(d,n)}{h}\le n^{dh} \]
  possible choices of $h$ members of $P_R(A)$, and there are no more
  than 
  \[
  h! 2^{h-1} C_{h-1}\leq h^{2h} \]
  Boolean combinations using unions and intersections of the $h$ sets, where
  $C_{h-1}$ is the $(h-1)^{\mathrm{th}}$ Catalan number
  ($C_i=\frac{1}{i+1}\binom{2i}{i}$). If 
  \[ 
  2^n> h^{2h}n^{dh} \geq |P_{R_h}(A)|\]
  then $A$ cannot be shattered. This inequality holds for $n\ge 3dh\log(dh)$
\end{proof}

\section{The VC-dimension of Classes of Queries}\label{sec:vcdimqueries}
In this section we develop a general bound on the VC-dimension of classes of
queries. We start by computing the VC-dimension of simple select queries on one
column and then move to more complex types of queries (multi-attributes select
queries, join queries). We then extend our bounds to general queries that are
combinations of multiple select and join operations.

\subsection{Select Queries}\label{sec:vcdimselqueries}
Let $\Tab$ be a table with $m$ columns $\Tab.C_1,\dotsc,\Tab.C_m$, and $n$
tuples. For a fixed column $\Tab.C$, consider the set $\Sigma_C$ of the
selection queries in the form 
\begin{equation}\label{eq:select}
\mbox{\texttt{SELECT }} *\mbox{\texttt{ FROM }} \Tab \mbox{\texttt{ WHERE }}
\Tab.C_i \op a \end{equation}
where $\op$ is an inequality operator (i.e., either ``$\ge$'' or
``$\le$'')\footnote{The operators
``$>$'' and ``$<$'' can be reduced to ``$\ge$'' and ``$\le$'' respectively.} and
$a\in D(\Tab.C)$. 

Let $q_1,q_2\in\Sigma_C$ be two queries. We say that $q_1$ is equivalent to
$q_2$ (and denote this fact as $q_1=q_2$) if their outputs are identical, i.e.,
they return the same set of tuples when they are run on the same database. Note
that $q_1=q_2$ defines a proper equivalence relation. 

Let $\Sigma^*_C \subseteq \Sigma_C$ be the maximum subset of $\Sigma_C$ that
contains no equivalent queries, i.e. it contains one query from each equivalent
class.
%
%


\begin{lemma}\label{lem:vcdimselmulcol}
  Let $\Tab$ be a table with $m$ columns $C_i$, $1\le i\le m$, and consider the
  set of queries 
  \[\Sigma^*_\Tab=\bigcup_{i=1}^m \Sigma^*_{C_i},\]
  where
  $\Sigma^*_{C_i}$ is defined as in the previous paragraph. Then, the range
  space $S=(\Tab,\Sigma^*_\Tab)$ has VC-dimension at most $m+1$.
\end{lemma}

\begin{proof}
  We can view the tuples of $\Tab$ as points in the $m$-dimensional space
  $\Delta=D(\Tab.C_1)\times D(\Tab.C_2) \times \cdots\times D(\Tab.C_m)$. A tuple
  $t\in\Tab$ such that $t.C_1=a_1, t.C_2=a_2, \ldots, t.C_m=a_m$ is represented
  on the space by the point $(a_1,a_2,\cdots,a_m)$.

  The queries in $\Sigma^*_\Tab$ can be seen as half spaces of $\Delta$. In
  particular any query in $\Sigma^*_\Tab$ is defined as in~\eqref{eq:select}
  and can be seen as the half space $\{(x_1,\cdots,x_i,\cdots,x_m) ~:~ x_j\in
  D(\Tab_j)\mbox{ for } j\neq i, \mbox{ and } x_i \op a\}\subseteq \Delta.$ It
  then follows from Lemma~\ref{lem:matousek} that $VC(S)\le m+1$.
\end{proof}

We now extend these result to general selection queries. Consider the set
$\Sigma^{2*}_\Tab$ of queries whose selection predicate can be expressed as the
Boolean combination of the selection predicates of at most two queries from
$\Sigma^*_\Tab$. These are the queries of the form:
\[
\mbox{\texttt{SELECT }} * \mbox{\texttt{ FROM }} \Tab \mbox{\texttt{ WHERE }}
\Tab.X_1 \op_1 a_1 \bool \Tab.X_2 \op_2 a_2
\]
where $\Tab.X_1$ and $\Tab.X_2$ are two columns from $\Tab$ (potentially,
$\Tab.X_1=\Tab.X_2$), $a_1\in D(\Tab.X_1)$, $a_2\in D(\Tab.X_2)$, ``$\op_i$'' is
either ``$\ge$'' or ``$\le$'' and ``$\bool$'' is either ``\texttt{AND}'' or
``\texttt{OR}''. Note that, in particular the queries in the form
\[
\mbox{\texttt{SELECT }} * \mbox{\texttt{ FROM }} \Tab \mbox{\texttt{ WHERE }} \Tab.X_1 \eqop a
\]
where $\eqop$ is either ``$=$'' or ``$\neq$'', belong to $\Sigma^{2*}_\Tab$
because we can rewrite a selection predicate containing one of these operators as
a selection predicate of two clauses using ``$\ge$'' and ``$\le$'' joined by
either $AND$ (in the case of ``$=$'') or $OR$ (in the case of ``$\neq$'').

By applying Lemma~\ref{lem:genboolcomp}, we have that the VC-dimension of the
range space $(\Tab,\Sigma^{2*}_\Tab)$ is at most $2(m+1)2\log((m+1)2)$, where
$m$ is the number of columns in the table $\Tab$.

We can generalize this result to $b$ Boolean combinations of selection
predicates as follows.

\begin{lemma}\label{lem:vcdimselgen}
  Let $\Tab$ be a table with $m$ columns, let $b>0$ and let $\Sigma^{b*}_\Tab$
  be the set of selection queries on $\Tab$ whose selection predicate is a
  Boolean combination of $b$ clauses. Then, the VC-dimension of the range space
  $S_b = (\Tab, \Sigma^{b*}_\Tab)$ is at most $3((m+1)b)\log((m+1)b)$.  
\end{lemma}

Note that we can not apply the bound
used in the proof of Lemma~\ref{lem:vcdimselmulcol} since
not all queries in $\Sigma^{b*}_\Tab$ are equivalent to axis-aligned
boxes. Once we apply Boolean operations on the outputs of the individual select
operation, the set of possible outputs, $S_b = (\Tab, \Sigma^{b*}_\Tab)$ , may
form complex subsets, including unions of disjoint (half-open) axis
aligned-rectangles and/or intersections of overlapping ones that cannot be
represented as a collection of half spaces. Thus, we need to apply a different technique here.

\begin{proof} 
 The output of a query $q$ in $\Sigma^{b^*}_\Tab$ can be seen as the Boolean
 combination (i.e. union and intersection) of the outputs of at most $b$
 "simple" select queries $q_i$ from $\Sigma^{*}_\Tab$ where each of these
 queries $q_i$ is as in~\eqref{eq:select}. An \texttt{AND} operation in the
 predicate of $q$ implies an intersection of the outputs of  the corresponding
 two queries $q_i$ and $q_j$, while an \texttt{OR} operation implies a union of
 the outputs. By applying Lemma~\ref{lem:genboolcomp}, we obtain the result.
\end{proof}

\subsection{Join Queries}\label{sec:vcdimjoinqueries}
Let $\Tab_1$ and $\Tab_2$ be two distinct tables, and let $R_1$ and $R_2$
be two families of (outputs of) select queries on the tuples of $\Tab_1$ and
$\Tab_2$ respectively. Let $S_1=(\Tab_1,R_1)$, $S_2=(\Tab_2,R_2)$ and let
$VC(S_1),VC(S_2)\geq 2$. Let $C$ be a column along which $\Tab_1$ and $\Tab_2$
are joined, and let $T_J=\Tab_1\times\Tab_2$ be the Cartesian product of the two
tables.

For a pair of queries $r_1\in R_1$, $r_2\in R_2$, let
\[
J^{\op}_{r_1,r_2}=\{(t_1,t_2) ~:~ t_1\in
r_1, t_2\in r_2, t_1.C \op t_2.C\},\]
where $\op\in\{>,<,\ge,\le,=,\neq\}$. $J^{\op}_{r_1,r_2}$ is the set of ordered pairs of
tuples (one from $\Tab_1$ and one from $\Tab_2$ that forms the output of the
join query 
\begin{equation}\label{eq:joinq}
\mbox{\texttt{SELECT }} *\mbox{\texttt{ FROM }} \Tab_1,\Tab_2 \mbox{\texttt{
WHERE }} r_1 \mbox{\texttt{ AND }} r_2.
\end{equation}
Here we simplify the notation by identifying select queries with
their predicates. We have $J^{\op}_{r_1,r_2}\subseteq
r_1\times r_2$ and $J^{\op}_{r_1,r_2}\subseteq T_J$. Let 
\[ J_C =\{J^{\op}_{r_1,r_2}~|~r_1\in R_1, r_2,R_2, \op\in\{>,<,\ge,\le,=,\neq\}
\}.\]
$J_C$ i the set of outputs of all join queries like the one in~\eqref{eq:joinq},
for all pairs of queries in $R_1\times R_2$ and all values of ``$\op$''. We
present here an upper bound to the VC-dimension of the range space
$S_J=(T_J,J_C)$.

\begin{lemma}\label{lem:vcdimjoin}
  $VC(S_J)\leq 3(VC(S_1)+VC(S_2))\log ((VC(S_1)+VC(S_2))).$
\end{lemma}

\begin{proof}
  Let $v_1=VC(S_1)$ and $v_2=VC(S_2)$. Assume that a set $A\subseteq T_J$ is
  shattered by $J_C$, and $|A|=v$.  Consider the two cross-sections $A_1 =
  \{x\in \Tab_1 ~:~ (x,y)\in A\}$ and $A_2 = \{y\in \Tab_2 ~:~ (x,y)\in A\}$.
  Note that $|A_1|\leq v$ and $|A_2|\leq v$ and by~\ref{corol:SauerProj}
  $|P_{R_1} (A_1)|\leq g(v_1,v)\le v^{v_1}$ and $|P_{R_2}
  (A_2)|\leq g(v_2,v) \le v^{v_2}$. For each set $r \in
  P_{J_C}(A)$ (i.e., for each subset $r\subseteq A$, given that
  $P_{J_C}(A)=2^A$)  there is a pair $(r_1,r_2)$, $r_1\in R_1$, $r_2\in R_2$,
  and there is $\op$, such that $r = A \cap J^{\op}_{r_1,r_2}$. Each of such
  pair $(r_1,r_2)$ identifies a distinct pair $(r_1\cap A_1, r_2\cap A_2) \in
  P_{R_1}(A_1)\times P_{R_2}(A_2)$, therefore each element of
  $P_{R_1}(A_1)\times P_{R_2}(A_2)$ can be identified at most $6$ times, the
  number of possible values for ``$\op$''. 

  In particular, for a fixed ``$\op$'', an element of $P_{R_1}(A_1)\times
  P_{R_2}(A_2)$ can be identified at most once.  To see this, consider two
  different sets $s_1, s_2 \in P_{J_C}(A)$. Let the pairs $(a_1,a_2)$,
  $(b_1,b_2)$, $a_1,b_1\in R_1$, $a_2,b_2\in R_2$, be such that $s_1 = A \cap
  J^{\op}_{a_1,a_2}$ and $s_2 = A \cap J^{\op}_{b_1,b_2}$. Suppose that $a_1\cap
  A_1 = b_1\cap A_1$  ($\in P_{R_1}(A_1)$) and $ a_2 \cap A_2 = b_2 \cap A_2$
  ($\in P_{R_2}(A_2)$). The set $s_1$ can be seen as $\{(t_1,t_2) ~:~ t_1\in
  a_1\cap A_1, t_2\in a_2\cap A_2 \mbox{ s.t.  } t_1.C \op t_2.C\}$. Analogously
  the set $s_2$ can be seen as $\{(t_1,t_2) ~:~ t_1\in b_1\cap A_1, t_2\in
  b_2\cap A_2 \mbox{ s.t. } t_1.C \op t_2.C\}$. But given that $a_1\cap A_1 =
  b_1\cap A_1$ and $ a_2 \cap A_2 = b_2 \cap A_2$, this leads to $s_1=s_2$, a
  contradiction. Hence, a pair $(c_1,c_2)$, $c_1\in P_{R_1}(A_1)$, $c_2\in
  P_{R_2}(A_2)$ can only be identified at most $6$ times, one for each possible
  value of ``$\op$''.  

  Thus,$ |P_{J_C} (A)|\leq 6|P_{R_1} (A_1)|\cdot |P_{R_2} (A_2)|.$ $A$ could not
  be shattered if $|P_{J_C}(A)|< 2^v$, i.e., if
  \[
  |P_{J_C}(A)|\leq 6|P_{R_1} (A_1)|\cdot |P_{R_2} (A_2)|\leq
  6g(v_1,v)g(v_2,v)\leq 6v^{v_1+v_2} < 2^v.
  \]
  The rightmost inequality holds for $v> 3(v_1+v_2)\log(v_1+v_2)$.
\end{proof}



The above results can be generalized to any query plan represented as a tree
where the select operations are in the leaves and all internal nodes are join
operations. As we said earlier, such a tree exists for any query.

\begin{lemma}\label{lem:vcdimmuljoin}
  Consider the class $Q$ of queries that can be seen as combinations of select
  and joins on $u>2$ tables $\Tab_1,\dots,\Tab_u$. Let $S_i=(\Tab_i,R_i)$,
  $i=1,\dots,u$ be the range space associated with the select queries on the $u$
  tables. Let $v_i=VC(S_i)$. Let $m$ be the maximum number of columns in a table
  $\Tab_i$. We assume $m\le \sum_i v_i$.\footnote{The assumption $m\le \sum_i
  v_i$ is reasonable for any practical case.} Let $S_Q = (\Tab_1\times\dots\times
  \Tab_u, R_Q)$ be the range space associated with the class $Q$. The range set
  $R_Q$ is defined as follows. Let $\rho = (r_1,\dots,r_u)$, $r_i\in R_i$, and
  let $\omega$ be a sequence of
  $u-1$ join conditions representing a possible way to join the $u$ tables $\Tab_i$,
  using the operators $\{>,<,\ge,\le,=,\neq\}$. We define the range 
  \[
  J^\omega_{\rho} = \{(t_1,\dots,t_u) ~:~ t_i\in r_i, \mbox{ s.t. }
  (t_1,\dots,t_u) \mbox{ satisfies } \omega\}.\]
  $R_Q$ is the set of all possible $J^\omega_{\rho}$. Then,
  \[
  VC(S_Q)\leq 4u(\sum_i VC(S_i))\log(u\sum_i VC(S_i)).
  \]
\end{lemma}

Note that this Lemma is not just an extension of Lemma~\ref{lem:vcdimjoin}
to join queries between multicolumns tables. Instead, it is an extension to
queries containing multiple joins (possibly between multicolumns tables).

\begin{proof}
Assume that a set $A\subseteq T_J$ is shattered by $R_Q$, and $|A|=v$. Consider
the cross-sections $A_i=\{x\in\Tab_i ~:~
(y_1,\dots,y_{i-1},x,y_{i+1},\dots,y_u)\in A\}, 1\le i\le u$.  Note that
$|A_i|\leq v$ and by~\ref{corol:SauerProj} $|P_{R_i} (A_i)|\leq g(v_i,v)
 \le v^{v_i}$. For each set $r \in P_{J_C}(A)$
(i.e., for each subset $r\subseteq A$, given that $P_{J_C}(A)=2^A$)  there is a
sequence $\rho=(r_1,\dots,r_u)$, $r_i\in R_i$, and there is an $\omega$, such
that $r = A \cap J^{\omega}_{\rho}$. Each sequence $\rho$ identifies a distinct
sequence $(r_1\cap A_1, r_2\cap A_2, \dots, r_u\cap A_u) \in
P_{R_1}(A_1)\times\dots\times P_{R_u}(A_u)$, therefore each element of
$P_{R_1}(A_1)\times \dots\times P_{R_u}(A_u)$ can be identified at most
$6^{u-1}$ times, one for each different $\omega$.

In particular, for a fixed $\omega$, an element of $P_{R_1}(A_1)\times
\dots\times P_{R_u}(A_u)$ can be identified at most once.
To see this, consider two different sets $s_1, s_2 \in
P_{J_C}(A)$. Let the vectors $\rho_a=(a_1,\dots,a_u)$,
$\rho_b=(b_1,\dots,b_u)$, $a_i,b_i\in R_i$, be such that $s_1 = A \cap
J^{\omega}_{\rho_a}$ and $s_2 = A \cap
J^{\omega}_{\rho_b}$. Suppose that $a_i\cap A_i = b_i\cap A_i$  ($\in
P_{R_i}(A_i)$). The set $s_1$ can be seen as $\{(t_1,\dots,t_u) ~:~ t_i\in a_i\cap
A_i,\mbox{ s.t. } (t_1,\dots,t_u) \mbox{ satisfies } \omega\}$. Analogously the
set $s_2$ can be seen as $\{(t_1,\dots, t_u) ~:~ t_i\in b_i\cap A_i, \mbox{ s.t.
} (t_1,\dots,t_u) \mbox{ satisfies } \omega\}$. But given that $a_i\cap A_i =
b_i\cap A_i$, this leads to $s_1=s_2$, a contradiction. Hence, a vector 
$(c_1,\dots,c_u)$, $c_i\in P_{R_i}(A_i)$, can only be identified at most $\ell$
times, one for each different $\omega$. For each of the $u-1$ join conditions
composing $\omega$ we need to choose a pair $(\Tab_1.A,\Tab_1.B)$ expressing the
columns along which the tuples should be joined. There are at most $g=\binom{um}{2}$
such pairs (some of them cannot actually be chosen, e.g. those of the type
$(\Tab_1.A, \Tab_1.B)$). There are then $\binom{g}{u-1}$ ways of choosing these
$u-1$ pairs. For each choice of $u-1$ pairs, there are $6^{(u-1)}$ ways of
choosing the operators in the join conditions ($6$ choices for $\op$ for each
pair). We have
\[
\ell\le \binom{\binom{um}{2}}{u-1}\cdot6^{(u-1)}\le (mu)^{2u}.
\]
Thus, $|P_{J_C} (A)|\leq \ell|P_{R_1} (A_1)|\cdot \dots \cdot|P_{R_u} (A_u)|$.
$A$ could not be shattered if $|P_{J_C}(A)|< 2^v$, i.e. if
\[
|P_{J_C} (A)| \leq \ell\cdot|P_{R_1} (A_1)|\cdot \dots\cdot |P_{R_u} (A_u)|\leq
\ell\cdot g(v_1,v)g(v_2,v)\dots g(v_u,v) \leq \leq (mu)^{2u} \cdot
v^{v_1+\dots+v_u}< 2^v.
\]
The rightmost inequality holds for $v> 4u\left(\sum_i v_i\right)\log (u\sum_i v_i)$.
\end{proof}

\subsection{General Queries}\label{sec:vcdimgenqueries}
Combining the above results we prove:
\begin{thm}\label{thm:vcdimgenqueries}
Consider a class $Q_{u,m,b}$ of all queries with up to $u-1$ join and $u$ select
operations, where each select operation involves no more than $m$ columns and $b$
Boolean operations, then 
\[
VC(Q_{u,m,b}) \leq
12u^2(m+1)b\log((m+1)b)\log(3u^2(m+1)b\log((m+1)b)).\]
\end{thm}

Note that Theorem~\ref{thm:vcdimgenqueries} gives an upper bound to the
VC-dimension. Our experiments suggest that in most cases the VC-dimension and
the corresponding minimum sample size are even smaller.

\section{Estimating Query Selectivity}\label{sec:applications}
We applying the theoretical result on the VC-dimension of queries 
to constructing a concrete
algorithm for selectivity estimation and query plan optimization.

\subsection{The general scheme}
Our goal is to apply Def.~\ref{defn:eapprox} and Thm.~\ref{thm:eapprox} to
compute an estimate of the selectivity of SQL queries. Let $Q_{u,m,b}$ be a
class of queries as in Theorem~\ref{thm:vcdimgenqueries}.
The class $Q_{u,m,b}$ defines a range space $S=(X,R)$ such that $X$ is the
Cartesian product of the tables involved in executing queries in $Q_{u,m,b}$,
and $R$ is the family of all output sets of queries in $Q_{u,m,b}$. 
Let $\Sam$ be an $\varepsilon$-approximation of $X$ and
let $r$ be the output set
of a query $q\in Q_{u,m,b}$ when executed on the original dataset, then $X\cap
r=r$ and $r\cap \Sam$ is the output set when the query is executed on the sample
(see details below). Thus, by Def.~\ref{defn:eapprox},
\[
\left|\frac{|X\cap r|}{|X|} - \frac{|\Sam\cap r|}{|\Sam|}\right|=
|\sigma_\Db(q)-\sigma_\Sam(q)| \le\varepsilon,
\]
i.e., the selectivity of running a query $q\in Q_{u,m,b}$ on an $\varepsilon$-approximation of $X$ 
is
within $\varepsilon$ of the selectivity of $q$ when executed on the original set of tables.
Note that for any execution plan of a query $q\in Q_{u,m,b}$, all the queries
that correspond to subtrees rooted at internal nodes of the plan are queries in
$Q_{u,m,b}$. Thus, by running query $q$ on an $\varepsilon$-approximation of $X$ we obtain accurate
estimates for the selectivity of all the subqueries defined by its execution
plan. Corresponding results are obtained by using a relative
  $(p,\varepsilon)$-approximation for $X$.
  
\subsection{Building and using the sample representation}\label{sec:building}
We apply Thm.~\ref{thm:eapprox} to probabilistically construct an $\varepsilon$-approximation of $X$.
A technical difficulty in algorithmic application of the theorem is that it is proven  for
a uniform sample of the Cartesian product
of all the tables in the database, while
in practice it is more efficient to maintain the table structure of the
original database in the sample. It is easier to sample each table
independently, and to run the query on a sample that consists of subsets of the
original tables rather than re-writing the query to run on a Cartesian product
of tuples. However, the Cartesian product of independent uniform samples of
tables is not a uniform sample of the Cartesian product of the
tables~\citep{ChaudhuriMN99}. We developed the following procedure to circumvent
this problem. Assume that we need a uniform sample of size $t$ from
$\mathbf{D}$, which is the Cartesian product of $\ell$ tables
$\Tab_1,\dotsc,\Tab_\ell$. We then sample $t$ tuples uniformly at random from
each table $\Tab_i$, to form a sample table $\Sam_i$. We add an attribute
$sampleindex$ to each $\Sam_i$ and we set the value in the added attribute for each tuple in
$\Sam_i$ to a unique value in $[1,t]$. Now, each sample table will contain $t$ tuples,
each tuple with a different index value in $[1,t]$. Given an index value
$i\in[1,t]$, consider the set of tuples $X_i=\{x_1,\dotsc,x_\ell\}$, $x_j\in\Sam_i$
such that $x_1.sampleindex = x_2.sampleindex =\dotsb=x_\ell.sampleindex=i$. $X_i$
can be seen as a tuple sampled from $\mathbf{D}$, and the set of all $X_i$,
$i\in[1,t]$ is a uniform random sample of size $t$ from $\mathbf{D}$. We run
queries on the sample tables, but in order to estimate the selectivity of a join
operation we count a tuple $Y$ in the result only if the set of tuples composing
$Y$ is a subset of $X_i$ for some $i\in[1,t]$. This is easily done by scanning
the results and checking the values in the $sampleindex$ columns (see Algorithms~\ref{alg:CreateSam} and~\ref{alg:ComputeSel}).

\begin{algorithm}[ht]
\SetKwInOut{Input}{input}
\SetKwInOut{Output}{output}
\SetKwFunction{drawRandomTuple}{drawRandomTuple}
\DontPrintSemicolon
\Input{sample size $s$, tables $\Tab_1,...,\Tab_k$.}
\Output{sample tables $\Sam_1,...,\Sam_k$ with $t$ tuples each.}
\For{$j\leftarrow 1$ \KwTo $k$}{
$\Sam_j\leftarrow\emptyset$\;
}
\For{$i\leftarrow 1$ \KwTo $s$}{
\For{$j\leftarrow 1$ \KwTo $k$}{
$t\leftarrow$ \drawRandomTuple{$\Tab_j$}\;
$s.sampleindex_j\leftarrow i$\;
$\Sam_j \leftarrow \Sam_j \cup \{t\}$\;
}
}
\caption{$\mathtt{CreateSample}(s,(T_1,\dots,T_k))$}\label{alg:CreateSam}
\end{algorithm}

\begin{algorithm}[ht]
\SetKwInOut{Input}{input}
\SetKwInOut{Output}{output}
\SetKwFunction{executeOperation}{executeOperation}
\Input{database operation $op$, sample database $S=(\Sam_1,\dots,\Sam_k)$ of
size $s$.}
\Output{the selectivity of $op$.}
$O_{op} \leftarrow \executeOperation{S, op}$\;
$(\ell_1,\dots,\ell_j)\leftarrow$ indexes of the sample tables involved in $op$ \;
$i\leftarrow 0$\;
\For{ $tuple\in O_{op}$}{
\lIf{$tuple.sampleindex_{\ell_1}=tuple.sampleindex_{\ell_2}=\dots=tuple.sampleindex_{\ell_j}$}{
$i\leftarrow i+1$\;
}
}
$selectivity \leftarrow i/s$\;
\caption{$\mathtt{ComputeSelectivity}(\Sam,op)$}\label{alg:ComputeSel}
\end{algorithm}

\begin{lemma}\label{lem:ComputeSel}
The $\mathtt{ComputeSelectivity}$ procedure (in Alg.~\ref{alg:ComputeSel})
executes a query on the Cartesian product of independent random samples of the
tables but outputs the selectivity that corresponds to executing the query on a
random sample of the Cartesian product of the original tables. 
\end{lemma}

\begin{proof}
The $CreateSample$ procedure chooses from each table a random sample of
$t$ tuples and adds to each sampled tuple an index in $[1,t]$. Each sample table has
exactly one tuple with each index value, and the Cartesian product of the sample
tables has exactly one element that is a concatenation of tuples, all with the
same index $i$ in their tables. Restricting the selectivity computation to these
$t$ elements (as in $ComputeSelectivity$) gives the result. 
\end{proof}

Note that our method circumvent the major difficulty pointed out
by Chaudhuri~et~al.~\citeyearpar{ChaudhuriMN99}. They also proved that, in general, it
is impossible to predict sample sizes for given two tables such that the join of
the samples of two tables will result in a sample of a required size out of the
join of the two tables. Our method does not require a sample of a given size
from the result of a join. The VC-dimension sampling technique requires only a
sample of a given size from the Cartesian product of the tables, which is
guaranteed by the above procedure.

Identifying the optimal query plan during query optimization may
require executing several candidate query plans on the sample. A standard
bottom-up candidate plan generation allows us to execute sub-plans once, store
their results and reuse them multiple times as they will be common to many
candidate plans. While the overhead of this execution-based selectivity
estimation approach will still likely be higher than that of pre-computation
based techniques (e.g., histograms),  the reduced execution times of highly
optimized plans enabled by better estimates, especially for complex and
long-running queries, will more than compensate for this overhead.   
Thus, storing intermediate results that are common to several executions will speed up
the total execution time on the sample. The significant improvement in the
selectivity estimates in complex queries well compensates for the extra work in
computing the selectivity estimates.

\section{Experiments}\label{sec:experiments}
This section presents the results of the experiments we run to validate our
theoretical results and to compare our selectivity estimation
method with standard techniques implemented in PostgreSQL and in
Microsoft SQL Server.

\paragraph{Goals.} The first goal of the experiments is to evaluate the practical
usefulness of our theoretical results. To assess this, we run queries on a large
database and on random samples of it of different sizes. We use the selectivity
of the each query in the random samples as an estimator for the selectivity in
the large database with the adjustments for join operations, as
described in the previous Section. We compute the error between the
estimate and the actual selectivity to show that the thesis of
Thm.~\ref{thm:eapprox} is indeed valid in practice. The use of a large number
of queries and of a variety of parameters allows us to evaluate the error
rate as a function of the sample size. We then compare our
method with the commonly used selectivity estimation based  on precomputed
histograms (briefly described in Sect.~\ref{sec:selhist}). We use histograms
with a different number of buckets to show that, no matter how fine-grained the
histograms might be, as soon as the inter-column and intra-bucket assumptions
are no longer satisfied, our approach gives better selectivity
estimates.

\subsection{Selectivity Estimation with Histograms}\label{sec:selhist}
In many modern database systems, the query optimizer relies on
histograms for computing data distribution statistics to help determine the most
efficient query plans. In particular, PostgreSQL 
uses one-dimensional equi-depth (i.e., equal frequency buckets) histograms and a
list of the most common values (MCV) for each column (of a database table) to
compute optimizer statistics. The MCV information stores the most
frequent $N$ items (by default $N=100$) and their frequency for each column. The
histograms (by default with $100$ bins) are built for the values not stored in
the MCV list. The selectivity of a constraint $A=x$, where $A$ is a
column and $x$ is a value is computed from the MCV list if $x$ is in the MCV
list or from the histogram bin that contains $x$ if $x$ is not in the MCV list.
The selectivity of a range constraint such as $A<x$ is computed with information from both the MCV
list and the histogram, i.e., the
frequencies of the most common values less than x and the frequency estimate for
$A<x$ from the histogram will be added to obtain the selectivity.

In PostgreSQL, the histograms and the MCV lists for the
columns of a table are built using a random sample of the tuples of the table. The
histograms and the MCV list for all columns of a table are based on the same
sample tuples (and are therefore correlated).  The sample size is computed for
each column using a formula based on the table size, histogram size, and a
target error probability developed by Chaudhuri~et~al.~\citeyearpar{ChaudhuriMN98}  and the largest sample size required by
the columns of a table is used to set the sample size of the table. 

Finally, the join selectivity of multiple constraints are computed using
the attribute independence assumption: e.g., selectivities are added in case of an
OR operator and multiplied for an AND operator.  Therefore, large selectivity
estimation errors are possible for complex queries and correlated inputs. 

\subsection{Setup}

\paragraph{Original tables.} The tables in our large database were randomly
generated and contain 20 million tuples each. There is a distinction between tables
used for running selection queries and tables used for running join (and
selection) queries. For tables on which we run selection queries only, the
distributions of values in the columns fall in two different categories:  
\begin{itemize}
  \item {\bf Uniform and Independent:} The values in the columns are chosen
    uniformly and independently at random from a fixed domain (the integer
    interval $[0,200000]$, the same for all columns). Each column is treated
    independently from the others. 
  \item {\bf Correlated:} Two columns of the tables contain values following a
    multivariate normal distribution with mean $M=\mu\mathbb{I}_{2,2}$ and a
    non-identity covariance matrix $\Sigma$ (i.e., the values in the two
    different columns are correlated). 
\end{itemize}
The tables for join queries should be considered in pairs $(A,B)$ (i.e., the
join happens along a common column $C$ of tables $A$ and $B$). The values in the
columns are chosen uniformly and independently at random from a fixed domain (the integer interval
$[0,200000]$, the same for all columns). Each column is treated independently
from the others. 

\paragraph{Sample tables.} We sampled tuples from the large tables uniformly,
independently, and with replacement, to build the sample tables. For the samples
of the tables used to run join queries, we drew random tuples uniformly at
random from the base tables independently and added a column $sampleindex$ to
each tuple such that each tuple drawn from the same base table has a different
value in the additional column and with tuples from different tables forming an
element of the sample (of the Cartesian product of the base tables) if they have
the same value in this additional column, as described in
Sect.~\ref{sec:building}.

For each table in the original database we create many sample tables of
different sizes. The sizes are either fixed arbitrarily or computed
using~\eqref{eq:eapprox} from Thm.~\ref{thm:eapprox}. The arbitrarily sized
sample tables contain between $10000$ and $1.5$ million tuples. To compute the
VC-dimension-dependent sample size, we fixed $\varepsilon=0.05$,
$\delta=0.05$, and $c=0.5$. The parameter $d$ was set to the best bound to the
VC-dimension of the range space of the queries we were running, as obtained from
our theoretical results. If we let $m$ be the number of columns involved in the
selection predicate of the queries and $b$ be the number of Boolean clauses in
the predicate, we have that $d$ depends directly on $m$ and $b$, as does the
sample size $s$ through \eqref{eq:eapprox} in Thm.~\ref{thm:eapprox}. For
selection queries, we used $m=1,2$ and $b=1,2,3,5,8$, with the addition of the
combination $m=5$, $b=5$. We run experiments on join queries only for some
combinations of $m$ and $b$ (i.e. for $m=1$ and $b=1,2,5,8$) due to the large
size of the resulting sample tables. Table~\ref{tab:samplesize} shows the sample
sizes, as number of tuples, for the combinations of parameters we used in our
experiments.

\begin{table}[ht]
  \centering
  \begin{tabular}{|c|c|cc|cc|}
    \cline{3-6}
    \multicolumn{2}{c}{$ $} &
    \multicolumn{2}{|c}{Select} &
    \multicolumn{2}{|c|}{Join} \\
    \hline
    $m$ & $b$ & VC-dim & Sample size & VC-dim & Sample size \\
    \hline
    \multirow{5}{*}{1} & 1 & 2 & 1000 & 4 & 1400\\
     & 2 & 4 & 1400 & 16 & 3800\\
     & 3 & 6 & 2800 & 36 & 7800 \\
     & 5 & 10 & 2600 & 100 & 20600\\
     & 8 & 16 & 3800 & 256 & 51800\\
    \hline
    \multirow{4}{*}{2} & 2 & 31 & 6800  & & \\
     & 3 & 57 & 12000 & & \\
     & 5 & 117 & 24000 & & \\
     & 8 & 220 & 44600 & & \\
    \hline
    5 & 5 & 294 & 59400 & & \\
    \hline
  \end{tabular}
  \caption{Sample Sizes}
  \label{tab:samplesize}
\end{table}
 

\paragraph{Histograms.} We built histograms with a different number of buckets,
ranging from $100$ to $10000$. Due to limitations in PostgreSQL, incrementing
the number of buckets in the histograms also increments the number of values
stored in the MCV list. Even if this fact should have a positive influence on
the quality of the selectivity estimates obtained from the histograms, our
results show that the impact is minimal, especially when the inter-column
independence and the intra-bucket uniformity assumptions are not satisfied.
For SQL Server, we built the standard single-column histograms and
computed the multi-column statistics which should help obtaining better
estimations when the values along the columns are correlated.

\paragraph{Queries.} 
For each combination of the parameters $m$ and $b$ and each large table (or pair
of large tables, in the case of join) we created $100$ queries, with selection
predicates involving $m$ columns and $b$ Boolean clauses. The parameters in each
clause, the range quantifiers, and the Boolean operators connecting the
different clauses were chosen uniformly at random to ensure a wide coverage of
possible queries.

\subsection{Results}
\paragraph{Selection Queries.} The first result of our experiments is that, for
all the queries we run, on all the sample tables, the estimate of the
selectivity computed using our method was within
$\varepsilon$ ($=0.05$) from the real selectivity. The same was not true for the
selectivity computed by the histograms. As an example, in the case of $m=2$,
$b=5$ and uniform independent values in the columns, the default PostgreSQL
histograms predicted a selectivity more than $\varepsilon$ off from the real
selectivity for 30 out of 100 queries.  Nevertheless, in some of cases the
histograms predicted a selectivity closer to the actual one than what our method
predicted. This is especially true when the histogram independence assumption
holds (e.g., for $m=2$, $b=5$ the default histograms gave a better prediction
than our technique in 11 out of 100 cases). Similar situations also arise for
SQLServer.

\begin{figure}[ht]
  \centering
  \includegraphics[scale=0.35]{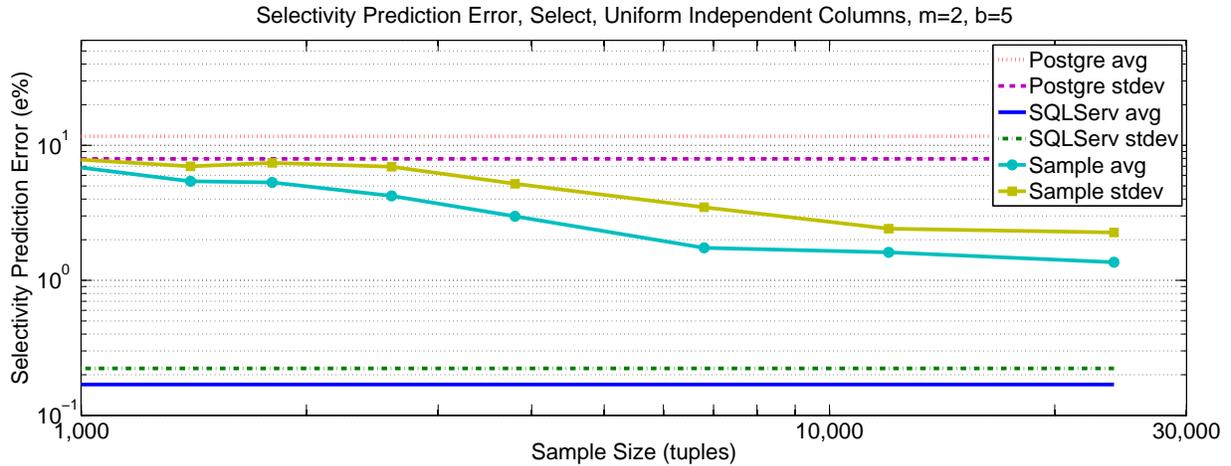}
  \caption{Select -- Uniform Independent Columns -- $m=2$, $b=5$}
  \label{fig:T_k5_u200k_unif_k2_b5_errperc}
\end{figure}

Since the selectivity estimated by the our method was always within $\varepsilon$
from the actual, we report the actual percent error, i.e. the quantity
$e_\%=\frac{100|p(\sigma_q)-\sigma_\Db(q)|}{\sigma_\Db(q)}$ where $p(\sigma_q)$
is the predicted selectivity. We analyze the average and the standard deviation
of this quantity on a set of queries and the evolution of these measures as
the sample size increases. We can see from
Fig.~\ref{fig:T_k5_u200k_unif_k2_b5_errperc}
and~\ref{fig:T_k2_correl_k2_b8_errperc} that both the average and the standard
deviation of the percentage error of the prediction obtained with our method
decrease as the sample size grows (the rightmost plotted
sample size is the one from Table~\ref{tab:samplesize}, i.e., the one computed 
in~Thm.\ref{thm:eapprox}. More interesting is the comparison in those figures between the performance of the
histograms and the performance of our techniques in predicting selectivities. When
the assumptions of the histograms hold, as is the case for the data plotted in
Fig.~\ref{fig:T_k5_u200k_unif_k2_b5_errperc}, the predictions obtained from the
histograms are good. 

\begin{figure}[ht]
  \centering
  \includegraphics[scale=0.35]{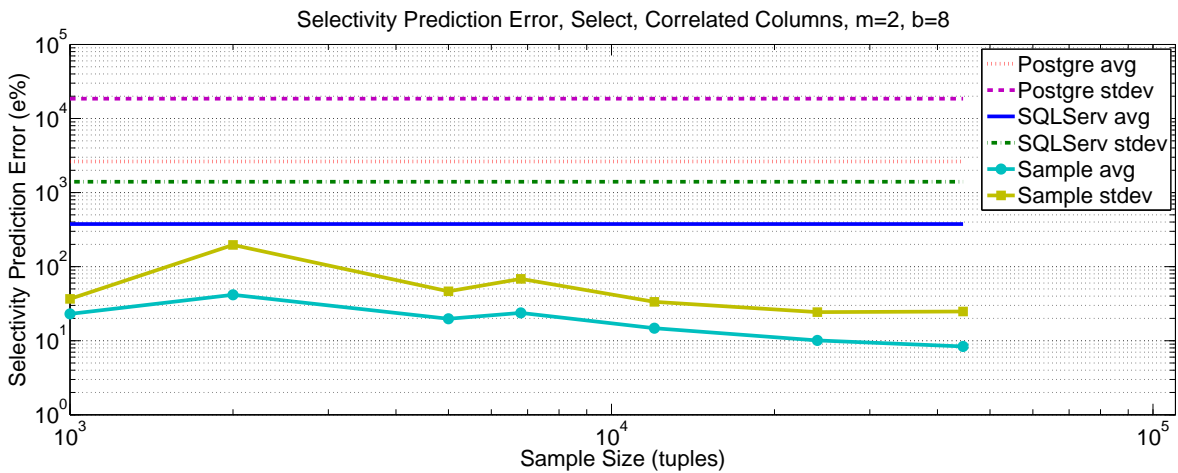}
  \caption{Select -- Correlated Columns -- $m=2$, $b=8$}
  \label{fig:T_k2_correl_k2_b8_errperc}
\end{figure}

But as soon as the data are
correlated (Fig.~\ref{fig:T_k2_correl_k2_b8_errperc}), our sampling method gives better
predictions than the histograms even at the smallest sample sizes and keeps
improving as the sample grows larger. It is also interesting to observe how the
standard deviation of the prediction error is much smaller for our method than
for the histograms, suggesting a much higher consistency in the quality of the
predictions. In Fig.~\ref{fig:T_k2_correl_k2_b8_errperc} we do not show multiple
curves for the different PostgreSQL histograms because increasing the number of
buckets had very marginal impact on the quality of the estimates, sometimes even
in the negative sense (i.e., an histogram with more buckets gave worse
predictions than an histogram with less buckets), a fact that can be explained
with the variance introduced by the sampling process used to create the
histograms. For the same reason we do not plot multiple lines for the
prediction obtained from the multi-columns and single-column statistics of SQL
Server: even when the multi-column statistics were supposed to help, as in the
case of correlated data, the obtained prediction were not much different from
the ones obtained from the single-column histograms.

\paragraph{Join Queries.} The strength of our method compared to histograms is
even more evident when we run join queries, even when the histograms independent assumptions are
satisfied. In our experiments, the predictions obtained using our technique were
always within $\varepsilon$ from the real values, even at the smallest sample
sizes, but the same was not true for histograms. For example, in the case of
$m=1$ and $b=5$, $135$ out of $300$ predictions from the histograms were more
than $\varepsilon$ off from the real selectivities.
Figure~\ref{fig:join_k1_b1_errperc} shows the comparison between the average and
the standard deviation of the percentage error, defined in the previous
paragraph, for the histograms and our method. The numbers include predictions
for the selection operations at the leaves of the query tree.

\begin{figure}[ht]
  \centering
  \includegraphics[scale=0.35]{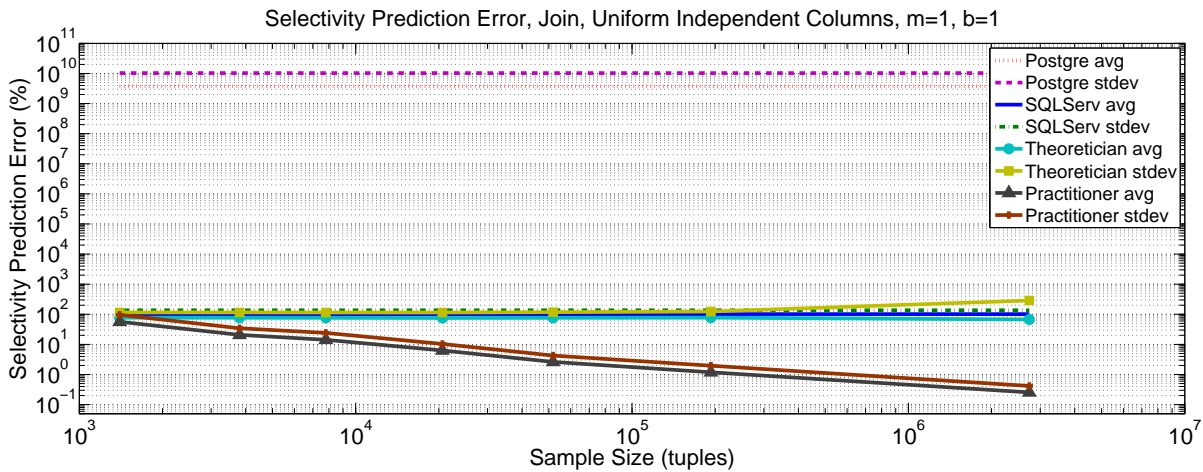}
  \caption{Join -- $m=1$, $b=1$}
  \label{fig:join_k1_b1_errperc}
\end{figure}

Again,
we did not plot multiple curves for histograms with a different number of
buckets because the quality of the predictions did not improve as the histograms
became more fine-grained. To understand the big discrepancy between the accurate
predictions of our method and the wrong estimates computed by the histograms in
PostgreSQL  we note that  for some join queries, the
histograms predicted an output size on the order of the hundreds of thousands
tuples but the actual output size was zero or a very small number of tuples.
Observing the curves of
the average and the standard deviation of the percentage error for the
prediction obtained with our method, we can see that at the smaller sample sizes
the quality of the predictions only improves minimally with the sample size.
This is due to the fact that at small sizes our prediction for the join
operation is very often zero or very close to zero, because the output of the
query does not contain enough pairs of tuples from the sample of the Cartesian
product of the input table (i.e. pairs of tuples with the same value in the
$sampleindex$ column). In these cases, the prediction can not be accurate at all
(i.e. the error is 100\% if the original output contained some tuples, or 0\% if
the query returned an empty set in the large databases). As soon as the sample
size grows more, we can see first a jump to higher values of the percentage
error, which then behaves as expected, i.e., decreasing as the sample size
increases. 

In Fig.~\ref{fig:join_k1_b1_errperc} we also show a
comparison between the percentage error of predictions obtained using our method
in two different ways: the ``theoretically correct'' way that makes use of the
number of pairs of tuples with the same value in the $sampleindex$ column and
the ``practitioner'' way which uses the size of the output of the join operation
in the sample, therefore ignoring the $sampleindex$ column. Recall that we had
to add the $sampleindex$ column because Thm.~\ref{thm:eapprox} requires a
uniform sample of the Cartesian product of the input tables.
As it is evident from Fig.~\ref{fig:join_k1_b1_errperc}, the ``practitioner''
way of predicting selectivity gives very good results at small sample sizes
(although it does not offer theoretical guarantees). These results are similar
in spirit, although not equivalent, to the theoretical conclusions presented by
Haas~et~al.~\citeyearpar{HaasNSS96} in the setting of selectivity estimation
using online sampling.

\section{Conclusions}\label{sec:compar}
We develop a novel method for estimating the selectivity of queries by executing
it on a concise, properly selected, sample of the database. We present a
rigorous analysis of our method and extensive experimental results demonstrating
its efficiency and the accuracy of its predictions.

Most commercial databases use histograms built on a single column, for selectivity
estimation. There has also been significant research on improving the estimate
using multidimensional
histograms~\citep{BrunoCG01,PoosalaI97,SrivastavaHMKT,WangS03} and join
synopses~\citep{AcharyaGPR99}. The main advantage of our method is that it gives
uniformly accurate estimates for the selectivity of any query within a predefined
VC-dimension range. Method that collect and store pre-computed statistics gives
accurate estimates only for the relations captured by the collected statistics,
while estimates of any other relation relies on an independence assumption.
, which give probabilistic guarantees on the error of the predicted
selectivity,%

To match the accuracy of our new method with histograms and  join synopses
one would need to create, for each table, a multidimensional histogram where the
number of dimensions is equal to the number of columns in the tables. The space
needed for a multidimensional histogram is exponential in the number of
dimensions, while the size of our sample representation is almost linear in that parameter. 
Furthermore, to estimate the selectivity for join operations
one would need to create join synopses for all pairs of columns in the database,
again in space that grows exponential in the number of columns.

It is interesting to note that the highly theoretical concept of VC-dimension
leads in this work to an efficient and practical tool for an important data
analysis problem.



\end{document}